\documentclass{lmcs}
\pdfoutput=1

\usepackage{lastpage}
\lmcsdoi{16}{3}{12}
\lmcsheading{}{\pageref{LastPage}}{}{}%
{May~01,~2019}{Aug.~20,~2020}{}


\usepackage{amsthm,amsmath,amssymb}
\usepackage[british]{babel}
\usepackage{mathtools}
\usepackage{stmaryrd}
\usepackage{tikz-cd}
\usepackage[utf8]{inputenc}

\input xypic

\newcommand{\mc}{\mathcal}
\newcommand{\msf}{\mathsf}
\newcommand{\mb}{\mathbb}
\newcommand{\graph}{\textup{graph}}

\DeclarePairedDelimiter{\pair}{\langle}{\rangle}
\DeclarePairedDelimiter{\pint}{[}{]}

\DeclarePairedDelimiter{\tint}{\llbracket}{\rrbracket}
\newcommand{\upset}{\uparrow\!}
\DeclareMathOperator{\id}{\textup{id}}

\begin{document}
\title{The Sierpinski Object in the Scott Realizability Topos}
\author{Tom de Jong}
\address{School of Computer Science, University of Birmingham}
\email{t.dejong@pgr.bham.ac.uk}

\author{Jaap van Oosten}
\address{Department of Mathematics, Utrecht University}
\email{j.vanoosten@uu.nl}
\thanks{Corresponding author: Jaap van Oosten}

\begin{abstract} We study the Sierpinski object $\Sigma$ in the realizability topos based on Scott's graph model of the $\lambda$-calculus. Our starting observation is that the object of realizers in this topos is the exponential $\Sigma ^N$, where $N$ is the natural numbers object. We define order-discrete objects by orthogonality to $\Sigma$. We show that the order-discrete objects form a reflective subcategory of the topos, and that many fundamental objects in higher-type arithmetic are order-discrete. Building on work by Lietz, we give some new results regarding the internal logic of the topos. Then we consider $\Sigma$ as a dominance; we explicitly construct the lift functor and characterize $\Sigma$-subobjects. Contrary to our expectations the dominance $\Sigma$ is not closed under unions. In the last section we build a model for homotopy theory, where the order-discrete objects are exactly those objects which only have constant paths.\end{abstract}
\maketitle
\section{Introduction} In this paper, we aim to revive interest in what we call the {\em Sierpinski object\/} in the {\em Scott realizability topos}. We show that it is of fundamental importance in studying the subcategory of {\em order-discrete\/} objects (section 3), arithmetic in the topos (section 4), the Sierpinski object as a dominance (section 5) and a notion of homotopy based on it (section 6). Section 2 deals with preliminaries and establishes notation.

\section{Preliminaries}
Independently observed by G.\ Plotkin and D.S.\ Scott (\cite{ScottD:lanl,ScottD:dattl,PlotkinG:settda}), there is a combinatory algebra structure (in fact, a $\lambda$-model structure) on the power set ${\mc P}(\mb{N})$ of the natural numbers.

Identifying ${\mc P}(\mb{N})$ with $\prod_{n\in\mb{N}}\{ 0,1\}$ and giving $\{ 0,1\}$ the Sierpinski topology (with $\{ 1\}$ the only nontrivial open set), we endow ${\mc P}(\mb{N})$ with the product topology. Concretely, basic opens of ${\mc P}(\mb{N})$ are of the form
$${\mc U}_p\; =\; \{ U\subseteq\mb{N}\, |\, p\subseteq U\}$$
for $p$ a finite subset of $\mb{N}$. A function $F:{\mc P}(\mb{N})\to {\mc P}(\mb{N})$ is continuous for this topology if and only if it has the property that $$n\in F(U)\text{ if and only if for some finite $p\subseteq U$, }n\in F(p).$$
Note, that every continuous function automatically preserves the inclusion $\subseteq$. 

Such continuous maps are encoded by elements of ${\mc P}(\mb{N})$ in the following way. Consider the following coding of finite subsets of $\mb{N}$:
$$e_n=\{ k_1,\ldots ,k_m\}\text{ if and only if }n=2^{k_1}+\cdots +2^{k_m}.$$
By convention, $e_0=\emptyset$. So: $e_1=\{ 0\}$, $e_2=\{ 1\}$, $e_3=\{ 0,1\}$, $e_{2^n}=\{ n\}$, etc.

Consider also a bijection $\langle {\cdot},{\cdot}\rangle :\mb{N}\times\mb{N}\to\mb{N}$, which we don't specify.

Now we define a binary operation (written as juxtaposition) on ${\mc P}(\mb{N})$:
$$UV=\{ n\, |\,\text{ for some }m, e_m\subseteq V\text{ and }\langle m,n\rangle\in U\}$$

For every function $F:{\mc P}(\mb{N})\to {\mc P}(\mb{N})$ we have
$${\sf graph}(F)\; =\; \{ \langle m,n\rangle\, |\, n\in F(e_m)\}$$
so that, when $F$ is continuous, ${\sf graph}(F)V=F(V)$ for all $V\in {\mc P}(\mb{N})$. Note that the operation $U,V\mapsto UV$ is a continuous map of two variables.

We think of the operation $U,V\mapsto UV$ as an {\em application operation}: apply the continuous function coded by $U$ to the argument $V$.

The structure ${\mc P}(\mb{N})$ together with this application function is denoted $\mc S$. The fact that $\mc S$ is a $\lambda$-model is worked out in detail in \cite{ScottD:dattl}. We shall use $\lambda$-terms to denote elements of $\mc S$ according to this interpretation. We use repeated juxtaposition and associate to the left when discussing iterated application: $VU_1\cdots U_n$ means $(\cdots ((VU_1)U_2)\cdots )U_n$. For continuous functions of several variables $F: {\mc P}(\mb{N})^n\to {\mc P}(\mb{N})$ we let
$${\sf graph}(F)\, =\, \{\langle m_1,\langle m_2,\langle\cdots m_n,k\rangle\cdots\rangle\rangle\, |\, k\in F(e_{m_1},\ldots ,,e_{m_n})\}$$
So with these conventions, ${\sf graph}(F)U_1\cdots U_n=F(U_1,\ldots ,U_n)$.

In $\mc S$, special combinators for important operations are available:
\begin{description}
\item[Pairing and unpairing combinators] ${\sf p},{\sf p}_0, {\sf p}_1\in {\mc S}$ satisfying ${\sf p}_0({\sf p}UV)=U$, ${\sf p}_1({\sf p}UV)=V$. In fact we can take
$$\{\langle 2^n,\langle 0,2n\rangle\rangle\, |\, n\in\mb{N}\}\cup\{\langle 0,\langle 2^n,2n+1\rangle\rangle\, |\, n\in\mb{N}\}$$
for $\sf p$, so that
$${\sf p}U\; =\; \{\langle 0,2n\rangle\, |\, n\in U\}\cup\{\langle 2^n,2n+1\rangle\, |\, n\in\mb{N}\}$$
and
$${\sf p}UV\; =\;\{ 2n\, |\, n\in U\}\cup\{ 2n+1\, |\, n\in V\} .$$
\item[Booleans] ${\sf t}=\{ 1\}$, ${\sf f}=\{ 0\}$. We have an ``if {\ldots } then {\ldots}, else{\ldots}'' operator $\sf q$, satisfying ${\sf qt}UV=U$ and ${\sf qf}UV=V$.
\item[Natural numbers] we write $\overline{n}$ for the element $\{ n\}$.
\item[Coding of tuples] we write $[U_0,\ldots ,U_n]$ for a standard coding of $n+1$-tuples in $\mc S$. This coding has the usual properties as given in \cite{OostenJ:reaics}: there are elements of $\mc S$ which give the length of a coded tuple, the $i$'th element (for appropriate $i$), and the code of the concatenation of two tuples.
\end{description}

As usual we have a {\em category of assemblies\/} over $\mc S$, denoted ${\sf Ass}({\mc S})$ and the {\em realizability topos\/} ${\sf RT}({\mc S})$, which we call the {\em Scott realizability topos} (unfortunately, the term ``Scott topos'' is already in use). We briefly introduce these categories.

An assembly over $\mc S$ is a pair $X = (|X|,E_X)$ where $|X|$ is a set and $E_X$ gives, for each $x\in |X|$, a nonempty subset $E_X(x)$ of $\mc S$. A morphism of assemblies $(|X|,E_X)\to (|Y|,E_Y)$ is a function $f:|X|\to |Y|$ for which there is some element $U$ of $\mc S$ which {\em tracks\/} $f$, meaning that for all $x\in |X|$ and all $V\in E_X(x)$ we have $UV\in E_Y(f(x))$.

The category ${\sf Ass}({\mc S})$ is a quasitopos. By way of example, let us construct exponentials explicitly: the exponent $(Y,E_Y)^{(X,E_X)}$ can be given as $(Z,E_Z)$ where $Z$ is the set of morphisms from $(X,E_X)$ to $(Y,E_Y)$ and $E_Z(f)$ is the set of those elements of $\mc S$ which track the morphism $f$. 

Actually, ${\sf Ass}({\mc S})$ has countable products: given a countable family $((X_n),E_n)_{n\in\mb{N}}$ of assemblies over $\mc S$, one has the assembly $(\prod_nX_n, E)$ where
$$E((x_n)_n)\; =\;\{U\, |\, U\overline{n}\in E_n(x_n)\text{ for all }n\} .$$

The category ${\sf Ass}({\mc S})$ has several subcategories of interest. First we recall that there is an adjunction
$${\rm Set}\stackrel{\nabla}{\to}{\sf Ass}({\mc S});\;\;\; {\sf Ass}({\mc S})\stackrel{\Gamma}{\to}{\rm Set};\;\;\; \Gamma\dashv \nabla$$
where $\Gamma (|X|,E_X)=|X|$, $\nabla (Y)=(Y, E)$ with $E(y)={\mc S}$ for all $y$.

An assembly $X$ is {\em partitioned\/} if $E_X(x)$ is always a singleton. It is not hard to see that the object $\nabla (Y)$ is isomorphic to $(Y,E)$ where $E(y)=\{\emptyset\}$ for all $y$, so all objects $\nabla (Y)$ are isomorphic to partitioned assemblies. Up to isomorphism, the partitioned assemblies are exactly the regular projective objects (the objects for which the representable functor into Sets preserves regular epimorphisms), which, in the topos ${\sf RT}({\mc S})$, coincide with the internally projective objects (the objects $X$ for which the endofuctor $(-)^X$ preserves epimorphisms); see \ref{characterizationinternallyprojective} below.

An assembly $X$ is {\em modest\/} if $E_X(x)\cap E_X(y)=\emptyset$ whenever $x\neq y$. The modest assemblies are, up to isomorphism, the ones for which the diagonal map $X\to X^{\nabla (2)}$ is an isomorphism (here, 2 denotes any two-element set); such objects are called {\em discrete\/} in \cite{HylandJ:disote}.

A special object of ${\sf Ass}({\mc S})$ is the {\em object of realizers\/} $S=({\mc S},E)$ where $E(U)=\{ U\}$; one sees that $S$ is both partitioned and modest. Every object which is both partitioned and modest is a regular subobject of $S$: that is, isomorphic to an object $({\mc A},E)$ where ${\mc A}\subseteq {\mc S}$ and $E$ is the restriction to $\mc A$ of the map defining $S$.

Another partitioned and modest object is the {\em natural numbers object\/} $N=(\mb{N},E_N)$, where $E_N(n)=\{\overline{n}\}$.

As Andrej Bauer observed in his thesis (\cite{BauerA:reaaca}), one can embed the category $T_0^{\rm cb}$ of countably based $T_0$-spaces into the subcategory of partitioned and modest assemblies over $\mc S$. To be precise, by ``countably based $T_0$-space'' we mean a $T_0$-space together with a chosen enumeration of a subbasis. If $(X,\{ S_n\, |\, n\in\mb{N}\} )$ is such, then the map $e:X\to {\mc S}$ given by
$$e(x)=\{ n\, |\, x\in S_n\}$$
defines an embedding of $X$ into $\mc S$. Moreover, given countably based spaces $X$ and $Y$ with associated embeddings $e:X\to {\mc S}, d:Y\to {\mc S}$, there is for every continuous map $f:X\to Y$ a continuous extension $F_f:{\mc S}\to {\mc S}$ satisfying $F_f{\circ}e=d{\circ}f$. This means that if we see the space $X$ with embedding $e:X\to {\mc S}$ as assembly $(X,E_X)$ where $E_X(x)=\{ e(x)\}$, then every continuous map $f:X\to Y$ is tracked (as a morphism of assemblies $(X,E_X)\to (Y,E_Y)$) by ${\sf graph}(F_f)$.

Finally, we have the {\em Scott realizability topos}, the realizability topos over $\mc S$, denoted ${\sf RT}({\mc S})$. Its objects are pairs $(X,\sim )$ where $X$ is a set and for $x,y\in X$, $[x\sim y]$ is a subset of $\mc S$ (so $\sim$ is a ${\mc P}({\mc S})$-valued relation on $X$), which satisfies certain conditions (we refer to \cite{OostenJ:reaics} for details). Every assembly $(X,E)$ is an object of ${\sf RT}({\mc S})$ if we define
$$[x\sim y]\; =\;\left\{\begin{array}{cl}E(x) & \text{if }x=y\\ \emptyset & \text{otherwise}\end{array}\right.$$
As subcategory of ${\sf RT}({\mc S})$, the category ${\sf Ass}({\mc S})$ is equivalent to the category of $\neg\neg$-stable objects of ${\sf RT}({\mc S})$, that is: the objects $(X,\sim )$ satisfying $\forall x,x':X.\neg\neg (x\sim x')\to x\sim x'$.
We shall not rehearse the structure of ${\sf RT}({\mc S})$ or its internal logic, based on realizability over $\mc S$, any further, as for this theory now standard texts are available. We only recall the following facts from \cite{OostenJ:reaics}, pp.135--7 (the proofs generalize to arbitrary pcas):
\begin{lem}\label{partitionedexponential}
    If $P$ is a partitioned assembly and $X$ is any object, then the exponential
    $X^P$ is isomorphic to the object $(|X|^{|P|},\approx)$ where 
    \[
        f \approx g = \pint{\forall p (E_P(p) \to f(p) \sim_X g(p))}.
    \]
\end{lem}
\begin{lem}\label{characterizationinternallyprojective}
    An object is internally projective in ${\sf RT}({\mc S})$ if and only if it is isomorphic to a
    partitioned assembly.
\end{lem}

The following small proposition will be of use later on. Given an assembly $(X,E)$ let ${\int}(X,E)$ be $(X',E')$ where $X'=\{ (x,U)\, |\, U\in E(x)\}$ and $E'(x,U)=\{ U\}$. Consider also ${\int\!\!\!\int}(X,E)=(X'',E'')$ where $X''=\{ (x,U,V)\, |\, U,V\in E(x)\}$ and $E''(x,U,V)=\{{\sf p}UV\}$. Clearly, we have a projection ${\int}(X,E)\to (X,E)$ and two projections ${\int\!\!\!\int}(X,E)\to {\int}(X,E)$, and these maps form a coequalizer diagram in ${\sf Ass}({\mc S})$.

\begin{prop}\label{assexponent} Let $(Y,\sim )$ be an arbitrary object of ${\sf RT}({\mc S})$ and let $(X,E)$ be an assembly over $\mc S$. The exponential $(Y,\sim )^{(X,E)}$ can be presented as follows: its underlying set is the set of functions from $\{ (x,U)\, |\, U\in E(x)\}$ to $Y$. Given two such functions $f,g$, the set $[f\approx g]$ consists of those coded triples $[P,Q,R]$ (coded in the sense of $\mc S$) of elements of $\mc S$, which satisfy the following three conditions:\begin{enumerate}
\item for $U\in E(x)$, $PU\in [f(x,U )\sim f(x,U )]$;
\item for $U,U'\in E(x)$, $QUU'\in [f(x,U)\sim f(x,U')]$;
\item for $U\in E(x)$, $RU\in [f(x,U)\sim g(x,U)]$.\end{enumerate}\end{prop}
\begin{proof} Observe that the objects ${\int}(X,E)$ and ${\int\!\!\!\int}(X,E)$ are partitioned assemblies, hence projective objects in ${\sf RT}({\mc S})$. Also note that the contravariant functor $(Y,\sim )^{(-)}$ sends coequalizers to equalizers. Finally, use Lemma~\ref{partitionedexponential}.\end{proof}

\section{The Sierpi\'nski object and order-discrete objects}
\begin{defi}
    The \emph{Sierpi\'nski object} $\Sigma$ is the image of Sierpi\'nski space
    under the embedding $T_0^{\text{cb}} \to \msf{RT}(\mc S)$. Concretely,
    $\Sigma$ is the modest set $(\{0,1\},E_\Sigma)$ with $$E_\Sigma(0) =
    \{\emptyset\} \text{ and } E(1) = \{\{1\}\} .$$
\end{defi}
\begin{defi}
    The \emph{object of realizers} $S$ is the modest set $(\mc S,\{-\})$. 
\end{defi}

The two objects above are related, as follows.
\begin{prop}\label{realizersisexponential}
    The object of realizers is isomorphic to the exponential $\Sigma^N$, where
    $N$ is the natural numbers object of $\msf{RT}(\mc S)$.
\end{prop}
\begin{proof}
    We prove that the underlying set of $\Sigma^N$ is $\{0,1\}^{\mb
    N}$. Since any morphism from $N$ to $\Sigma$ is in particular a function
    from $\mb N$ to $\{0,1\}$, one inclusion is clear. Conversely, if $f\colon
    \mb N\to \{0,1\}$, then $f$ is tracked by $\graph(F)$, where $F\colon \mc S
    \to \mc S$ is the continuous function defined by:
    \[
        \overline n \mapsto 
        \begin{cases}
            \overline 1 &\text{if } f(n) = 1; \\
            \emptyset   &\text{if } f(n) = 0.
        \end{cases}
    \]
    Thus, $\Sigma^N$ is the assembly $(\{0,1\}^{\mb N}, E)$ where $E(f)\subseteq
    \mc S$ is the set of trackers of the function $f$ (note, that this set is always nonempty). 
    
    We have a canonical bijection
    \[    
        \chi \colon \mc S \to \{0,1\}^{\mb N}, \quad U \mapsto \chi_U 
    \]
    (where $\chi_U$ is the characteristic function of $U$) with inverse
    \[
        \chi^{-1}\colon \{0,1\}^{\mb N} \to \mc S, \quad f \mapsto
        f^{-1}(\{1\}).
    \]
    
    It remains to prove that these functions are tracked. For $\chi$, consider
    the continuous function $F\colon \mc S \to \mc S$ given by
    $
        U \mapsto \{\pair{2^n,1}\mid n\in U\}.
    $
    Then $\chi$ is tracked by graph($F$). Indeed, if $U\in\mc S$ and $n\in\mb
    N$, then $\graph(F)U\overline n = F(U)\overline n$ and this is $\overline 1$
    if $n\in U$ and $\emptyset$ otherwise.

    For the inverse of $\chi$, we define $G\colon \mc S \to \mc S$ continuous by
    $U\mapsto \{n\in\mb N \mid U\overline n = \overline 1\}$. (This is
    continuous, because the application of $\mc S$ is continuous.) We claim that
    $\chi^{-1}$ is tracked by $\graph(G)$.
    Indeed, for $f \in \{0,1\}^\mb N$ and $U \in E(f)$, we have $n \in G(U)$ if
    and only if $f(n) = 1$, since $U$ tracks $f$.
\end{proof}

\subsection{Order-discrete objects}
\begin{defi}
    An object $X$ of $\msf{RT}(\mc S)$ is \emph{order-discrete} if it is
    orthogonal to $\Sigma$, viz.\ the diagonal $X \xrightarrow{\delta_X} X^\Sigma$
    is an isomorphism.
\end{defi}

We wish to characterize the order-discrete objects in terms of their realizers.
To this end, we first present some useful lemmas.
\begin{lem}\label{deltaepic}
    An object $X$ is order-discrete if and only if the diagonal $\delta_X$ is an
    epimorphism.
\end{lem}
\begin{proof}
    Of course, $\delta_X$ is epic if $X$ is order-discrete. Conversely, suppose
    $\delta_X$ is an epi. Clearly, the unique morphism $!_\Sigma \colon \Sigma
    \to 1$ is an epimorphism. Hence, the diagonal $\delta_X = X^{!_\Sigma}$ is
    monic. Thus, $\delta_X$ is an iso, as desired.
\end{proof}
\begin{prop}\label{characterizationorder-discrete}
    An object $X$ is order-discrete if and only if there is $A\in\mc S$ such that
    for any $x,x' \in |X|$: if $U\in [x\sim_X x]$ and $V\in [x' \sim_X x']$ with
    $U\subseteq V$, then $AUV \in [x\sim_X x']$.
\end{prop}
\begin{proof}
    Suppose first that $X$ is order-discrete. We construct the desired element
    $A \in \mc S$. Assume we have $x,x'\in |X|$ and $U\in [x\sim_X x], V\in [x'
    \sim_X x']$ with $U\subseteq V$. 
    By Lemma~\ref{partitionedexponential} we have $X^\Sigma \cong
    (|X|^{\{0,1\}},\approx)$. Define $f \colon \{0,1\} \to |X|$ by $f(0) = x$
    and $f(1) = x'$. The map $H \colon \mc S^2 \to \mc S$ defined as 
    \[
        (W,W') \mapsto \{\pair{0,n} \mid n \in W\} \cup \{\pair{2,n} \mid n \in
        W' \}
    \]
    is easily seen to be continuous. Write $G = \graph(H)$. We claim that $GUV
    \in [f \approx f]$. Indeed, $GUV = \{\pair{0,n} \mid n \in U\} \cup
    \{\pair{2,n} \mid n \in V\}$, so that $GUV \emptyset = U \in [f(0) \sim_X
    f(0)]$. Moreover, $GUV \overline 1 = U \cup V = V \in [f(1) \sim_X f(1)]$, since
    $e_0 = \emptyset$ and $e_2 = \overline 1$ and $U\subseteq V$. Thus, $GUV \in
    [f \approx f]$. Now let $R \in \mc S$ be a realizer of the fact that
    $\delta$ is epic. Then $R(GUV)$ is an element of $[\forall p (E_\Sigma(p)
    \to x_0 \sim_X f(p))]$ for some $x_0 \in |X|$. Thus, $R(GUV)\emptyset \in
    [x_0 \sim_X x]$ and $R(GUV)\overline 1 \in [x_0 \sim_X x']$. Finally, let
    $t,s\in \mc S$ respectively realize transitivity and symmetry of $\sim_X$.
    Then, we see that
    \[
        \lambda {u v}. t (\msf p (s (R(G u v)\emptyset))(R(G u v)\overline 1))
    \]
    is the desired element $A$.

    Conversely, suppose we have an $A\in\mc S$ as in the proposition. Let $f\colon
    \{0,1\}\to |X|$ be arbitrary. By Lemma~\ref{deltaepic}, it suffices to show that
    from an element of $[f \approx f]$, we can continuously find an $x\in |X|$
    and an element of $[\forall p(E_\Sigma(p) \to x \sim_X f(p))]$. Let $F\in
    [f\approx f]$. Then $F\emptyset \in [f(0) \sim_X f(0)]$, $F\overline 1 \in
    [f(1) \sim_X f(1)]$ and $F\emptyset \subseteq F\overline 1$, so
    $A(F\emptyset)(F\overline 1) \in [f(0) \sim_X f(1)]$ and
    $A(F\emptyset)(F\emptyset) \in [f(0) \sim_X f(0)]$. Hence, if we set $x=f(0)$,
    then the graph of the continuous function 
    \[
        \emptyset \mapsto A(F\emptyset)(F\emptyset),\quad W\neq\emptyset \mapsto
        A(F\emptyset)(F\overline 1)
    \]
    is the desired element. 
\end{proof}
\begin{cor}\label{orderdiscreteassembly}
    An assembly $X$ is order-discrete if and only if the existence of realizers
    $U\in E_X(x)$ and $V\in E_X(y)$ with $U\subseteq V$ implies that $x$ and $y$
    are equal.
\end{cor}
\begin{proof}
    Immediate.
\end{proof}
\begin{exa}
    The natural numbers object $N$ (c.f.\ the proof of
    \ref{realizersisexponential}) is easily seen to be order-discrete using the
    corollary above.
\end{exa}

Recall from the section on preliminaries that an object is called
\emph{discrete} if it is orthogonal to $\nabla(2)$. For an assembly $X$, this
means that the realizing sets are disjoint. Thus, Corollary~\ref{orderdiscreteassembly} shows that any
order-discrete assembly is also discrete (alternatively, one might observe that the map from $\Sigma$ to $\nabla (2)$ is an epi in ${\sf Ass}({\mc S})$, and therefore the map $X^{\nabla (2)}\to X^{\Sigma}$ is monic. From this one also easily deduces that order-discrete implies discrete, for assemblies). In fact, this holds for all objects,
not just assemblies.
\begin{prop}
    Any order-discrete object is discrete.
\end{prop}
\begin{proof}
    This follows from Proposition~\ref{characterizationorder-discrete} and \cite{OostenJ:reaics}, Corollary~3.2.20. The statement there carries over
    to arbitrary realizability toposes.
\end{proof}
\begin{rem}
    One might expect the availability of a simpler proof of the proposition
    above, i.e.\ a proof that does not rely on the characterization of the
    (order-)discrete objects. However, note that although we have a morphism
    $\Sigma \to \nabla(2)$, it is not a regular epimorphism, so one cannot apply
    \cite{HylandJ:disote}, Lemma 2.2.
\end{rem}

In \cite{HylandJ:disote}, it it shown that the class of objects orthogonal to an
object $A$ satisfies some nice properties if $A$ is well-supported and
internally projective. By Lemma~\ref{characterizationinternallyprojective}, the object $\Sigma$ is internally projective. It is clearly well-supported, which gives us the following proposition:
\begin{prop}\label{odreflective}
    The order-discrete objects are closed under all existing limits, subobjects
    and quotients in $\msf{RT}(\mc S)$. Moreover, they form an exponential
    ideal. Finally, they form a reflective subcategory of $\msf{RT}(\mc S)$.
\end{prop}
\begin{proof}
    The first statement is true, because any collection of orthogonal objects is
    closed under all existing limits (see \cite{HylandJ:disote}, p. 1). Similarly,
    any collection of orthogonal objects forms an exponential ideal. 
    The other closure properties follow from Lemmas 2.3 and 2.8 of \cite{HylandJ:disote}, and Lemma~\ref{characterizationinternallyprojective} above. The final claim is proven as a
    corollary right after Lemma 2.3 of \cite{HylandJ:disote}.
\end{proof}

\begin{rem}\label{odreflection} For comparison with the notion of homotopy to be defined below, we give a concrete representation of the order-discrete reflection of an arbitrary object $(X,\sim )$. It has the same underlying set $X$, and we let $[x\approx y]$ consist of those coded (in the sense of $\mc S$) $(3n+1)$-tuples $[U_0,V_0,W_0,\ldots ,U_{n-1},V_{n-1},W_{n-1},U_n]$ for which there is a sequence $(x_0,y_0,\ldots ,x_n,y_n)$ of elements of $X$, such that the following conditions hold:\begin{enumerate}\item $x_0=x$ and $y_n=y$;
\item $U_k\in [x_k\sim y_k]$ for $0\leq k\leq n$;
\item $V_k\in [y_k\sim y_k]$ and $W_k\in [x_{k+1}\sim x_{k+1}]$ for $0\leq k <n$;
\item $V_k\subseteq W_k$ or $W_k\subseteq V_k$ must hold.\end{enumerate}\end{rem}

\section{Arithmetic in ${\sf RT}({\mc S})$}
The order-discrete modest sets already make a disguised appearance in
\cite{LietzP:comrpo}. Herein, Lietz singles out a class of so-called
well-behaved modest sets. It is not hard to show that a modest set is
well-behaved if and only if it is order-discrete and moreover, it has the
join-property, which we define now.
\begin{defi}
    An assembly has the \emph{join-property} if it is isomorphic to an
    assembly $X$ whose realizing sets are closed under binary joins, viz.\ if
    $U \in E_X(x)$ and $V \in E_X(x)$, then $U \cup V \in E_X(x)$ for all $x \in
    |X|$.
\end{defi}
\begin{exa}
    The natural numbers object $N$ has the join-property.
\end{exa}
The following proposition is a very slight generalization of \cite{LietzP:comrpo}, Theorem
2.3.
\begin{prop}
    Let $X$ and $Y$ be modest sets. If $X$ has the join-property and $Y$ is
    order-discrete, then the following form of the Axiom of Choice holds in
    $\msf{RT}(\mc S)$\textup{:}
    \[
        \forall x{:}X\exists y{:}Y \phi(x,y) \to \exists f{:}Y^X \forall x{:}X
        \phi(x,f(x)).
    \]
\end{prop}
\begin{proof}
    See the proof of \cite{LietzP:comrpo}, Theorem 2.3.
\end{proof}
We have already seen that the order-discrete objects form an exponential ideal
and are closed under binary products. In \cite{LietzP:comrpo}, Theorem 2.2, it is
shown that this also holds for the modest objects that have the join-property.
Since $N$ is order-discrete and has the join-property, it follows that the Axiom
of Choice holds for all finite types. Consequently, the following two principles
\begin{alignat*}{2}
    &\textup{}\quad &&\forall f{:} N^N \exists x{:} N \phi(f,x) \to \forall
    f{:} N^N \exists xy {:} N \forall g{:} N^N (\overline f y = \overline g y
    \to \phi(g,x)), \\
    & &&\text{where $\overline f y = \overline g y$ is short for $\forall z {:}
    N (z < y \to f(z) = g(z))$} \\
    & &&\textup{(\emph{Weak Continuity for Numbers})} \\
    &\textup{}\quad &&\forall F {:} N^{(N^N)} \forall f {:} N^N \exists x
    {:} N \forall g{:} N^N (\overline f x = \overline g x \to F(f) = F(g)) \\
    & &&\textup{(\emph{Brouwer's Principle})}
\end{alignat*}
are equivalent (and false) in $\msf{RT}(\mc S)$; c.f.\
\cite{LietzP:comrpo}, Section 2.3.

\subsection{Second-order arithmetic in ${\sf RT}({\mc S})$}
We close this section by some brief considerations on second-order arithmetic in
realizability toposes and $\msf{RT}(\mc S)$ in particular. 
\begin{lem}\label{simplification}
    Let $A$ be any partial combinatory algebra. Write $\overline n$ for the
    $n$-th Curry numeral in $A$ and let ${\msf{p}}$ be a pairing combinator in
    $A$ with projections ${\msf{p}}_0$ and ${\msf{p}}_1$.
    
    The power object ${\mc P} (N)$ of
    the natural numbers object $N$ in ${\sf RT}(A)$ can be given as the object $(P,\approx)$
    where $P$ is the set
    \[
        P = \{\phi \colon \mb N \to \mc P(A) \mid \text{for all $n \in \mb
        N$, if $a \in \phi(n)$, then ${\msf{p}}_1 a = \overline n$}\}
    \]
    and
    \[
        \psi \approx \psi = [\forall n(\phi(n) \leftrightarrow \psi(n))].
    \]
\end{lem}
\begin{proof}
    Straightforward.
\end{proof}

We now turn our attention to $\lnot\lnot$-stable subsets of $N$, viz.\ subsets
$X$ such that $$\forall x(\lnot\lnot (x \in X) \to x \in X)$$ is true in
$\msf{RT}(A)$. We will abbreviate this formula by
$\textup{Stab}(X)$.

We formulate the following triviality as a proposition for comparison with
\ref{stablenotuniform} below, in the case of $A = {\mc S}$.

\begin{prop}\label{noteverythingisstable}
    If $A$ is nontrivial, then not every subset is $\lnot\lnot$-stable, i.e.\ the sentence $\forall X
    \textup{Stab}(X)$ is not valid in ${\sf RT}(A)$.
\end{prop}
\begin{proof} Note that $\forall X
    \textup{Stab}(X)$ implies $\forall p.\neg\neg p\to p$ (where $p$ is a variable of type $\Omega$), which fails in ${\sf RT}(A)$ as soon as $A$ has two disjoint nonempty subsets. If $A$ is trivial then ${\sf RT}(A)={\rm Set}$, and $\forall X
    \textup{Stab}(X)$ obviously holds.
    \end{proof}
    
By contrast, in ${\sf RT}({\mc S})$ we have the following result (substituting $\mc S$ for $A$ in the definition of $P$ in Lemma~\ref{simplification}):
\begin{prop}\label{stablenotuniform}
    For any $\alpha \in P$, we have a realizer of $\textup{Stab}(\alpha)$.
    Hence, the sentence $\forall X \lnot\lnot \textup{Stab}(X)$ is true in
    $\msf{RT}(\mc S)$. 
\end{prop}
\begin{proof}
    Let $\alpha \in P$. Note that there is a function $f_{\alpha}\colon\mb N\to
    \mc S$ such that 
    \[
        f_{\alpha}(n) \in 
        \begin{cases}
            [n\in \alpha] &\text{if } [n \in \alpha]\neq\emptyset; \\
            \{\emptyset\} &\text{else}.
        \end{cases}
    \]
    Moreover, there is an element $F_{\alpha}\in\mc S$ such that
    $F_{\alpha}\overline n = f_{\alpha}(n)$. Now note that $\lambda {xy}. F_\alpha
    x$ realizes $\textup{Stab}(\alpha)$.
\end{proof}
\begin{cor}\label{fo=true} First-order arithmetic in ${\sf RT}({\mc S})$ is the same as first-order arithmetic in Set.\end{cor}
\begin{proof} Given a true first-order arithmetical sentence $\phi$, we have, by Proposition~\ref{stablenotuniform}, a realizer for $\forall x(\psi (x)\vee\neg\psi (x))$, for every subformula $\psi$ of $\phi$. From this one constructs a realizer for~$\phi$.\end{proof}

Just as for ordinary natural numbers, we write $\langle -,\!- \rangle$ for a
definable coding of $N^2$ to $N$ in ${\sf RT}(A)$. The sentence $\forall
X\exists Y(\textup{Stab}(Y)\land \forall x(x\in X \leftrightarrow \exists y
\langle{y,x}\rangle\in Y))$ is known as \emph{Shanin's Principle\/}. It holds in the Effective Topos (\cite{OostenJ:reaics}, p.\ 127). Internally, it says that every
subset of $N$ is covered by a stable subset of $N$.
\begin{prop}\label{shanindoesnothold}
    If $|A| > |\mb N|$, then Shanin's Principle does not hold in
    ${\sf RT}(A)$.
\end{prop}
\begin{proof}
    Assume for the sake of contradiction that it does. Then we have a realizer
    \[
        R \in \bigcap_{\alpha \in P}\bigcup_{\beta \in P}\tint{
        \textup{Stab}(\beta) \land \forall x(x\in \alpha \leftrightarrow
        \exists y \langle y,x \rangle \in \beta)}.
    \]
    Let us write $R_0 = \msf{p_0} R$ and $R_1 = \msf{p_1} R$.
    For each $a \in A$, define the element $\alpha_a \in P$ by:
    \[
        \alpha_a(0) = \{\msf p a \overline 0\}\quad\text{and}\quad \alpha_a(n+1) =
        \emptyset\quad\text{for any } n \in \mb N.
    \]
    For each $a\in A$, pick some $\beta_a \in P$ such that $R \in \tint{
    \textup{Stab}(\beta_a) \land \forall x (x \in \alpha_a \leftrightarrow
    \exists y \langle y,x \rangle \in \beta_a)}$.

    From $R_1$, we effectively obtain $R' \in A$ such that for every $a \in
     A$, we have:
    \[
        R'(\msf p a \overline 0) \in [\langle m, 0\rangle \in \beta_a]
    \]
    for some $m \in \mb N$.
    As $|A|>|\mb N|$, there must be two different $a,a'\in A$ such that 
    \begin{equation*}
        R'(\msf p a \overline 0) \in [\langle m,0\rangle \in \beta_a] \text{ and }
        R'(\msf p a' \overline 0) \in [\langle m,0\rangle \in \beta_{a'}] \tag{$\ast$}
    \end{equation*}
    for the same $m\in\mb N$.
    
    From $R_0$, we effectively obtain $s \in A$ witnessing the stability of
    $\beta_a$ and $\beta_{a'}$. Thus, by $(\ast)$, we can use $s$ to get a
    common realizer:
    \[
        s \overline{\langle m, 0\rangle}\msf i \in [\langle m,0\rangle \in \beta_a]
        \cap [\langle m, 0\rangle \in \beta_{a'}].
    \]

    Finally, using $s\overline{\langle m,0\rangle}\msf i$ and $R_1$, we find a
    realizer in the intersection \break $[0 \in \alpha_a] \cap [0 \in
    \alpha_{a'}] = \alpha_a \cap \alpha_{a'}$. But this is impossible, because
    $a$ and $a'$ are different, so that $\alpha_a$ and $\alpha_{a'}$ are disjoint.
\end{proof}
\begin{cor}
    Shanin's Principle does not hold in $\msf{RT}(\mc S)$ and
    $\msf{RT}({\mc K}_2)$ (where ${\mc K}_2$ is Kleene's second model).
\end{cor}

\section{The Sierpi\'nski object as a dominance}
One can also study the Sierpi\'nski object from a synthetic domain theoretic
standpoint. This was done by Wesley Phoa (\cite{PhoaW:domtrt}, Chapter 12 and \cite{PhoaW:bdgm}, Proposition 3.1), who worked in the category of modest sets over the r.e.\ graph model, John Longley (\cite{LongleyJ:reatls}, 5.3.7) and Alex Simpson (\cite{LongleyJ:uadtrm}). Our treatment is based on Proposition~\ref{realizersisexponential}.
\begin{defi}\label{membershiponS}
    Define a relation $\in$ between $N$ and $S$ (the object of realizers) by
    taking the following pullback
    \[
        \begin{tikzcd}
            \in \ar[dr, "\lrcorner", phantom, very near start] \ar[rr] \ar[d] &
                                                                              &
            1 \ar[d, "t"] \\ 
            N \times S \cong N \times \Sigma^N \ar[r, "\textup{ev}"] & \Sigma
                \ar[r]
                                                             & \Omega
        \end{tikzcd}
    \]
    where $\Sigma\to\Omega$ is the morphism induced by the function $0 \mapsto
    \emptyset$ and $1 \mapsto \mc S$.
\end{defi}
\begin{rem}\label{inisnotnotstable}
    Observe that $\in$ is given by
    \[
        [n \in U] =
        \begin{cases}
            \msf p \overline n U &\text{if $n$ is an element of $U$}; \\
            \emptyset &\text{else}.
        \end{cases}
    \]
    In particular, $\in$ is $\lnot\lnot$-stable.
\end{rem}

\begin{lem}\label{IP_R}
    Let $\mc A$ be a arbitrary \emph{total} pca and let $\msf{RT}(\mc A)$ be its
    realizability topos. Let $A = (\mc A, \{-\})$ be the object of realizers in
    the topos. The scheme
    \[
        (\textup{IP}_A) \quad (\lnot\phi \to \exists x{:}A \psi) \to \exists x{:}A
        (\lnot\phi \to \psi)
    \]
    with $x$ not free in $\phi$ is valid in $\msf{RT}(\mc A)$.
\end{lem}
\begin{proof}
    It is not hard to verify that  $\lambda u.{\msf{p}}({\msf{p}}_0(u{\sf k}))(\lambda v.{\msf{p}}_1(u{\sf k}
    ))$ realizes the scheme.
\end{proof}

\begin{prop}\label{omega'def}
    The subobject $\Omega'$ of $\Omega$ given by
    
    \[
        \Omega ' = \tint{\exists X{:}S (p \leftrightarrow 1 \in X)}
    \]
    with $p$ ranging over $\Omega$,
    is a dominance in $\msf{RT}(\mc S)$. Moreover, it is $\lnot\lnot$-separated,
    i.e.\ $\forall p{:}\Omega' (\lnot\lnot p \to p)$. Furthermore, $\bot \in \Omega'$.
\end{prop}
\begin{proof}
    Double negation separation is immediate by Remark~\ref{inisnotnotstable}. Further,
    it is clear that $\top \in \Omega'$ (take $X = \overline 1$). 
    
    Suppose $p\in\Omega '$ and $p\to (q\in\Omega ')$. Take $U\in\mc S$ such that $p
    \leftrightarrow 1 \in U$. Then, $1 \in U \to \exists X{:}S (q
    \leftrightarrow 1\in X)$, so by Lemma~\ref{IP_R} and Remark~\ref{inisnotnotstable}, we
    get that $\exists X{:}S (1 \in U \to (q \leftrightarrow 1 \in X))$. Take such
    $V\in\mc S$. Then, $p\land q \leftrightarrow 1\in U\cap V$. Hence, $p\land q
    \in \Omega'$ and $\Omega'$ is a dominance.
\end{proof}

We have already remarked that the object $\Omega '$ from \ref{omega'def} is
$\lnot\lnot$-separated. Indeed, it is isomorphic to an assembly.
\begin{prop}\label{omega'isomorphictosigma}
    The object $\Omega'$ is isomorphic to $\Sigma$.
\end{prop}
\begin{proof}
    First of all, observe that $\Omega'$ is isomorphic to the object
    $(\mc P(\mc S),\sim)$ where
    \begin{align*}
        \mc U \sim \mc V &= \mc U \leftrightarrow \mc V \land E(\mc U),
        \text{with} \\
        E(\mc U) &= \{[W,U] \in \mc S \mid \text{if $1\in W$, then $U \in \mc U$
        and if $\mc U \neq \emptyset$, then $1 \in W$} \}.
    \end{align*}
    Define a continuous function $F\colon \mc S \to \mc S$ by
    \[
        F(\emptyset) = \emptyset \text{ and } F(V) = \overline 1 \text{ for any
        non-empty $V$}.
    \]
    Next, define $\Phi \colon \mc P(\mc S) \times \{0,1\} \to \mc P(\mc S)$ by
    \[
        (\mc U, i) \mapsto \{[W,U,C] \mid [W,U] \in E(\mc U), C \in E_\Sigma(i)
        \text{ and } i = 1 \Leftrightarrow 1 \in W \}.
    \]
    We show that $\Phi$ is a functional relation from $(\mc P(\mc S),\sim)$ to
    $\Sigma$. Strictness is immediate. For single-valuedness, suppose we have
    $[W,U,C] \in \Phi(\mc U,i)$ and $[W',U',C'] \in \Phi(\mc U,j)$. We show that
    $i = j$. By definition, we have
    \[
        i = 1 \Leftrightarrow 1 \in W \Leftrightarrow U \in \mc U \Rightarrow
        1 \in W' \Leftrightarrow j = 1
    \]
    and similarly, $i = 0 \Rightarrow j=0$. Thus, $i=j$, as desired.
    Suppose $[W,U,C] \in \Phi(\mc U,i)$ and $B\in [\mc U \leftrightarrow \mc
    V]$. We must effectively obtain an element of $\Phi(\mc V,i)$. But if $B_0$
    realizes $\mc U \to \mc V$, then one easily sees that $[W,B_0 U,C]$ is an
    element of $\Phi(\mc V,i)$. So, $\Phi$ is relational. For totality, suppose
    $[W,U]\in E(\mc U)$, then $[W,U,F(W)] \in \Phi(\mc U,i)$ for some $i\in
    \{0,1\}$, by construction of $E$ and $F$. We conclude that $\Phi$ is a
    functional relation.

    Moreover, (the arrow represented by) $\Phi$ is easily seen to be epic. For,
    if $C\in E_\Sigma(i)$, then $[F(C),F(C),C] \in \Phi(\mc U_i,i)$ with $\mc
    U_1 = \{\overline 1\}$ and $\mc U_0 = \emptyset$ by construction of $F$ and
    definition of $E_\Sigma$. 

    Finally, we prove that $\Phi$ is monic and hence that $\Phi$ represents an
    isomorphism, as desired. Suppose we have $[W,U,C] \in \Phi(\mc U,i)$ and
    $[W',U',C] \in \Phi(\mc V,i)$. It suffices to effectively provide an element
    of $\mc U \leftrightarrow \mc V$, since $[W,U]$ is an element of $E(\mc U)$
    already. But $[\lambda x.U',\lambda x.U]$ is easily seen to do the job.
\end{proof}

Proposition \ref{omega'isomorphictosigma} implies that we have a notion of $\Sigma$-{\em subobjects\/} of any object in $\msf{RT}(\mc S)$: a mono $Y\to X$ represents a $\Sigma$-subobject of $X$ if and only if its classifying map factors through $\Sigma$. We write $U\subseteq _{\Sigma}X$ to indicate that $U$ is a $\Sigma$-subobject of $X$.

The following proposition characterizes these $\Sigma$-subobjects for assemblies $X$.
\begin{prop}\label{characterizationsigmasubassemblies}
    Let $X$ be an assembly. There is a bijective correspondence between
    morphisms $X \to \Sigma$ and subsets $X'\subseteq |X|$ for which there is
    an open $\mc U\subseteq\mc S$ with the following properties:
    \begin{align*}
        x \in X'     &\Rightarrow E_X(x) \subseteq \mc U; \\
        x \not\in X' &\Rightarrow E_X(x) \cap \mc U = \emptyset. 
    \end{align*}
    Moreover, an assembly $Y$ is a $\Sigma$-subobject of $X$ if and only if $Y$
    is isomorphic to some assembly $(X',E)$ where $X'\subseteq |X|$ is as above
    and $E$ is the restriction of $E_X$ to $X'$.
\end{prop}
\begin{proof}
    Let $f$ be a morphism from $X$ to $\Sigma$ that is tracked by $U\in\mc S$. Set
    \[
        X' = \{x\in X \mid f(x) = 1\} \quad\text{and}\quad \mc U = \{V\in\mc S \mid
            UV = \overline 1\}.
    \]
    We show that $\mc U$ is open.
    Let $\mathcal Q \coloneqq \{p\subseteq\mc S \mid p \text{ is finite and } Up
        = \overline 1\}$. Recall the notation $\upset p = \{V \in \mc S \mid p
    \subseteq V\}$. By continuity of the application, one can show that $\mc U =
    \bigcup_{p \in \mathcal Q} \upset p$. Thus, $\mc U$ is an open of $\mc S$.
    
    
    From the definition of $\mc U$ and the fact that $U$ tracks $f$, it is
    immediate that $\mc U$ has the desired properties.

    For the converse, assume we are given an open $\mc U\subseteq \mc S$ and a
    subset $X'\subseteq |X|$ with the properties stated. Define $f\colon X \to
    \Sigma$ by $f(x) = 1$ if $x\in X'$ and $f(x) = 0$ if $x\not\in X'$. We
    claim that it is tracked by $\graph(F)$ where $F(U) = \{1 \mid U \in \mc U\}$.
    That this $F$ is continuous follows from the assumption that $\mc U$ is open.
    Now if $x\in X$ and $U\in E_X(x)$, then either $f(x) = 1$, in which case
    $E_X(x) \subseteq \mc U$, so that $F(U) = \overline 1\in E_\Sigma(f(x))$; or
    $f(x) = 0$, in which case $E(x) \cap \mc U = \emptyset$, so that $F(U) =
    \emptyset\in E_\Sigma(f(x))$. So $f$ is tracked, as desired.

    That the operations above are each other's inverses is readily verified. The
    final claim follows immediately from the construction above and the
    description of pullbacks in the category of assemblies.
\end{proof}

The referee has suggested that 
Proposition~\ref{characterizationsigmasubassemblies} gives a notion of {\em synthetic topology}
on $\msf{RT}(\mc{S})$ in the sense of \cite{EscardoM:tophil}. However, observe the caveat after Proposition~\ref{Omega'notclosedunderor}.

For any dominance $\Sigma\subset\Omega$ in a topos, one has for every object $X$ the partial map classifier for $\Sigma$-subobjects: that is, a $\Sigma$-subobject $\eta _X:X\to L(X)$ which classifies partial maps into $X$ defined on a $\Sigma$-subobject. It follows that $L$ is a functor (generally  dubbed the ``lift functor'') and $\eta$ is a natural transformation from the identity functor to $L$. Now Propositions \ref{omega'def} and \ref{omega'isomorphictosigma} together imply that the functor $L$ restricts to an endofunctor on the category of assemblies. The following proposition describes this restriction explicitly.

\begin{prop}\label{liftfunctor}
    The lift functor $L$ on ${\sf Ass}({\mc S})$ is given by on objects by:
    \[
        L(X) = \left(|X|\cup\{\bot_X\},E_{L(X)}\right),
    \]
    where $\bot_X$ is some element not in $|X|$ and
    \[
        E_{L(X)}(\bot_X) = \{\emptyset\}\quad\text{and}\quad E_{L(X)}(x) =
        \{[U,\overline 1] \mid U \in E_X(x)\} \text{ for $x \in |X|$}.
    \]
    Given an arrow $f\colon X \to Y$, we define $L(f)$ as the unique extension
    of $f$ satisfying $\bot_X \mapsto \bot_Y$.
    The natural transformation $\eta \colon \id_{\msf{Ass}(\mc S)} \to L$ is
    defined as $\eta_X(x) = x$.
\end{prop}
\begin{proof}
    Given a morphism $f\colon X \to Y$ of $\msf{Asm}(\mc S)$ tracked by $U_f\in
    \mc S$, note that $L(f)$ is tracked, as
    \[
        [V_0,V_1] \mapsto \begin{cases}
            [U_f V_0, \overline 1] &\text{if } 1 \in V_1; \\
            \emptyset &\text{else};
        \end{cases}
    \]
    is a continuous map $\mc S\to\mc S$. Verifying that $L$ is indeed a functor
    is routine. Also, note that $\eta_X$ is tracked by $\lambda {u}. [u,\overline
    1]$. That $\eta$ is natural is easily checked.

    Finally, suppose we have a morphism $f\colon U \to Y$ and $U
    \subseteq_\Sigma X$. By Proposition~\ref{characterizationsigmasubassemblies}, we may
    assume that we have $|U|\subseteq |X|$ and an open $\mc U$ such that for
    $x\in |X|$: 
    \[
        \text{if } x\in |U|, \text{then } E_X(x) \subseteq \mc U
        \quad\text{and}\quad \text{if } x\not\in |U|, \text{then } E_X(x)\cap
        \mc U = \emptyset. 
    \]
    Define $\tilde f\colon X \to L(Y)$ by
    \[
        \tilde f(x) = 
        \begin{cases}
            f(x) &\text{if } x\in |U|; \\
            \bot_Y &\text{else}.
        \end{cases}
    \]
    Note that $\tilde f$ is tracked, for if $f$ is tracked by $U_f$, then
    \[
        U \mapsto
        \begin{cases}
            U_f U &\text{if } U \in \mc U; \\
            \emptyset &\text{else}
        \end{cases}
    \]
    is a continuous map $\mc S\to\mc S$, because $\mc U$ is open.

    If we have morphisms $g\colon Z \to X$ and $h\colon Z \to Y$ such that
    $\tilde f g = \eta_Y h$, then we must have that $g(z) \in |U|$ for any $z\in
    |Z|$. Hence, $g$ factors uniquely through $(U,E_U)$. This proves that
    $\tilde f$ makes the square 
    $$\diagram U\rto^{\subseteq _{\Sigma}}\dto_f & X\dto^{\tilde{f}}\\ Y\rto_{\eta _Y} & L(Y)\enddiagram$$
    into a pullback.
    
    It remains to show that it is unique with this property. To this end,
    suppose we have $f'\colon X \to L(Y)$, such that the square 
    $$\diagram U\rto^{\subseteq _{\Sigma}}\dto_f & X\dto^{f'}\\ Y\rto_{\eta _Y} & L(Y)\enddiagram$$
    is a pullback.
    From the commutativity of the square, it follows that for $x\in |U|$ we
    must have $f'(x) = f(x) = \tilde f(x)$. Now suppose for a contradiction that
    we have $x_0\in |X|\setminus |U|$ and $f'(x_0) \in |Y|$. The universal
    property of the pullback then yields a map $U\cup\{x_0\} \to U$ such that
    the inclusion $|U|\cup\{x_0\} \to |X|$ factors through $|U|$, but this is
    impossible, as $x_0\not\in |U|$. Hence, no such $x_0$ exists and therefore,
    $f'$ and $\tilde f$ coincide.
\end{proof}
\begin{rem}
    Observe that $L(1) \cong \Sigma$ and that $\eta_1 = 1 \xrightarrow{t}
    \Sigma$.
\end{rem}

\begin{lem}\label{liftsigmasubobject}
    For any $X$, we have $X\subseteq_\Sigma L(X)$ via $\eta_X$.
\end{lem}
\begin{proof}
    It is straightforward to verify that
    \[
        \begin{tikzcd}
            X \arrow[r] \arrow[d, "\eta_X"] & 1 \arrow[d, "t"] \\
            L(X) \arrow[r,"\chi_X"] & \Sigma
        \end{tikzcd}
    \]
    with $\chi_X(x) = 1$ and $\chi_X(\bot_X) = 0$ is a pullback. (Note the map
    $F(V) = \overline 1$ if $1 \in \mathsf {p_1} V$ and $\emptyset$
    otherwise is continuous and its graph tracks $\chi_X$.)
\end{proof}

As always, the lift functor $L$ is a monad on
${\sf Ass}({\mc S})$. The multiplication $\mu\colon L^2(X) \to L(X)$ is given by
the map $x\mapsto x, \bot_X \mapsto \bot_X, \bot_{L(X)} \mapsto \bot_X$.

\section{A notion of homotopy for the order-discrete objects}
In \cite{OostenJ:nothet}, a notion of homotopy for the Effective Topos is given. This notion is tied up with the subcategory of discrete objects, in the following way: internally, for each object $X$, the ``discrete reflection'' of $X$ is the set of path components of $X$. It is just natural to wonder whether this would also work for the order-discrete objects in ${\sf RT}({\mc S})$.

Another interpretation of homotopy for the Effective Topos is in the paper \cite{FruminD:homtmf} by Dan Frumin and Benno van den Berg. They build a very elegant framework on purely category-theoretic notions. The basic ingredients are: an interval object $\mb{I}$ with two distinct points $\partial _0,\partial _1: 1\to\mb{I}$ , binary maps $\wedge ,\vee :I\times I\to I$ satisfying some axioms, and a dominance $\Sigma$ which is closed under finite unions and satisfies the condition that the inclusion $[\partial _0,\partial _1]:1+1\to\mb{I}$ is a $\Sigma$-subobject. From these notions, a subcategory of ``fibrant objects'' is defined, and there is a Quillen model structure on this subcategory. When applying to the Effective Topos they take the trivial dominance: $\Omega$ itself.

In ${\sf RT}({\mc S})$ we have an interesting dominance: the Sierpi\'nski object. However, the Frumin-van den Berg framework is not applicable for this dominance, as the following proposition shows.

\begin{prop}\label{Omega'notclosedunderor} The Sierpi\'nski object, as subobject of $\Omega$, is not closed under finite unions.\end{prop}
\begin{proof} We apply the characterization \ref{omega'isomorphictosigma}. Let the variables $X,Y,Z$ range over the object of realizers $S$ (which in turn is seen as subobject of ${\mc P}(N)$). We prove that the statement
  \begin{equation}
  \forall XY\exists Z(1\in Z\leftrightarrow 1\in X\vee 1\in Y)
  \tag{$\ast$}
  \end{equation}
is not valid in ${\sf RT}({\mc S})$. Indeed, working out the realizability of $(\ast )$, we see that we need continuous functions $F$ and $G$ on $\mc S$ such that for all $U,V\in {\mc S}$ we have:\begin{itemize}
\item If $1\in U$ or $1\in V$ then $1\in F(U,V)$.
\item If $1\in F(U,V)$ then either $G(F(U,V))=\overline{0}$ and $1\in U$, or $G(F(U,V))=\overline{1}$ and $1\in V$.\end{itemize}
However, this is impossible: consider $G(F(\{ 1\} ,\{ 1\} ))$. If this is equal to $\overline{0}$ then by continuity we cannot have $G(F(\emptyset ,\{ 1\} ))=\overline{1}$; the other case is similar.\end{proof}

Caveat: note that if we interpret Escard\'o's ``synthetic topology'' in ${\sf RT}({\mc S})$, we see from Proposition~\ref{Omega'notclosedunderor} that this topology will not be standard, but ``substandard'' in the terminology of \cite{EscardoM:tophil}.

In view of Proposition~\ref{Omega'notclosedunderor}, we revert to the more `hands-on' approach of \cite{OostenJ:nothet}. The main extra complication is the essential asymmetry of the object $\Sigma$. In the effective topos case, we had for every natural number $n$ a `generic interval' $I_n$ (which essentially is an iterated pushout of copies of $\nabla (2)$). Here, we have more generic intervals, indexed by 01-sequences.
\begin{defi}\label{genericintervaldef} For each natural number $n\geq 1$ and each 01-sequence $\sigma =(\sigma _0,\ldots ,\sigma _{n-1})$ we define the assembly $I_{n,\sigma}=(\{ 0,\ldots ,n\} ,E_{n,\sigma})$ where:
$$\begin{array}{lcl}
E_{n,\sigma}(0) & = & \left\{\begin{array}{cl} \{ {\sf p}\overline{0}\,\overline{1},{\sf p}\overline{1}\emptyset\} & \text{if }\sigma _0=0;\\
\{ {\sf p}\overline{0}\,\overline{1},{\sf p}\overline{1}\,\overline{1}\} & \text{if }\sigma _0=1; \end{array}\right. 
\\
E_{n,\sigma}(k+1)\text{ (for $k+1<n$)} & = & \left\{\begin{array}{cl} \{ {\sf p}\overline{k}\,\overline{1},{\sf p}\overline{k+1}\emptyset\} & \text{if }\sigma _k=0,\sigma _{k+1}=0;\\
\{ {\sf p}\overline{k}\,\overline{1},{\sf p}\overline{k+1}\,\overline{1}\} & \text{if }\sigma _k=0,\sigma _{k+1}=1;\\
\{ {\sf p}\overline{k}\emptyset ,{\sf p}\overline{k+1}\emptyset\} & \text{if }\sigma _k=1,\sigma _{k+1}=0;\\
\{ {\sf p}\overline{k}\emptyset ,{\sf p}\overline{k+1}\,\overline{1}\} & \text{if }\sigma _k=1,\sigma _{k+1}=1;
\end{array}\right.
\\
E_{n,\sigma}(n) & = & \left\{\begin{array}{cl} \{ {\sf p}\overline{n}\,\overline{1},{\sf p}\overline{n+1}\emptyset\} & \text{if }\sigma _{n-1}=0;\\
\{ {\sf p}\overline{n}\emptyset ,{\sf p}\overline{n+1}\emptyset\} & \text{if }\sigma _{n-1}=1.\end{array}\right.
\end{array}$$
So, each set $E_{n,\sigma}(k)$ consists of two elements $\alpha _k$, $\beta _k$ which satisfy the following conditions:\begin{itemize}
\item ${\sf p}_0\alpha _k=\overline{k}$, ${\sf p}_0\beta _k=\overline{k+1}$;
\item $\beta _k\subset\alpha _{k+1}$ if $\sigma _k=0$;
\item $\beta _k\supset\alpha _{k+1}$ if $\sigma _k=1$.\end{itemize}
The object $I_{n,\sigma}$ is called a {\em generic interval of length $n$}.
 \end{defi}
\begin{prop}\label{exponentInsigma} Let $(X,\sim )$ be an arbitrary object of ${\sf RT}({\mc S})$ and let $I_{n,\sigma}$ be a generic interval of length $n$. The exponent $(X,\sim )^{I_{n,\sigma}}$ can be described as follows: its underlying set is the set of $2n+2$-tuples $(\vec{x},\vec{y})=(x_0,y_0,\ldots ,x_n,y_n)$ of elements of $X$. Given such a tuple, we define $E(\vec{x},\vec{y})$ as the set of coded $3n-2$-tuples (coded in the sense of $\mc S$)
$$[U_0,V_0,W_0,\ldots ,U_{n-1},V_{n-1},W_{n-1},U_n]$$
of elements of $\mc S$ such that we have:\begin{itemize}
\item $U_k\in [x_k\sim y_k]$ for $0\leq k\leq n$;
\item $V_k\in [y_k\sim y_k]$, $W_k\in [x_{k+1}\sim x_{k+1}]$ for $0\leq k<n$;
\item $V_k\subseteq W_k$ if $\sigma _k=0$, and $V_k\supseteq W_k$ if $\sigma _k=1$.\end{itemize}
For two such tuples $(\vec{x},\vec{y}),(\vec{x'},\vec{y'})$, the set $[(\vec{x},\vec{y})\approx (\vec{x'},\vec{y'})]$ is the set of coded triples (in the sense of $\mc S$) $[\alpha,\beta ,\gamma ]$ of elements of $\mc S$, which satisfy:\begin{itemize}
\item $\alpha\in E(\vec{x},\vec{y})$;
\item $\beta\in E(\vec{x'},\vec{y'})$;
\item $\beta\overline{k}\in [x_k\sim x'_k]$ for $0\leq k\leq n$.\end{itemize}
\end{prop}
\begin{proof} This is left to the reader, who may find it useful to review Proposition~\ref{assexponent}. It should be noted that realizers of elements of $(X,\sim )^{I_{n,\sigma}}$ don't contain information about $\sigma$. This is not necessary, since for any $V,W\in {\mc S}$ we have a continuous endomap on $\mc S$ which sends ${\sf p}\overline{k}\emptyset$ to $V\cap W$ and ${\sf p}\overline{k}\,\overline{1}$ to $V\cup W$, so it is enough to know that some inclusion between $U_k,V_k$ exists. Note also that the number $n$ can be retrieved from any element of $E(\vec{x},\vec{y})$.\end{proof}
\begin{defi}\label{oepp} As in \cite{OostenJ:nothet}, we say that an arrow $I_{n,\sigma}\to I_{m,\tau}$ is {\em order and end-point preserving\/} if its underlying function: $\{0,\ldots ,n\}\to\{ 0,\ldots ,m\}$ is so (with respect to the usual order).\end{defi}
\begin{defi}\label{pathobjectdef} Let $(X,\sim )$ be an object of ${\sf RT}({\mc S})$. Its {\em path object\/} ${\sf P}(X,\sim )$ defined as follows. Its underlying set is the set of all tuples $(x_0,y_0,\ldots ,x_n,y_n)$, which we write as $(\phi ,n)$. For two such tuples $(\phi ,n)$ and $(\psi ,m)$ we let $[(\phi ,n)\sim (\psi ,m)]$ be the set of coded triples (in the sense of $\mc S$) $[\alpha ,s,\beta ]$ of elements of $\mc S$ such that for suitable 01-sequences $\sigma ,\tau$ we have:\begin{enumerate}
\item $\alpha\in E(\phi ,n)$ as element of $(X,\sim )^{I_{n,\sigma}}$;
\item $\beta\in E(\psi ,m)$ as element of $(X,\sim )^{I_{m,\tau}}$;
\item There are order and endpoint preserving maps $I_{p,\rho}\stackrel{f}{\to}I_{n,\sigma}$, $I_{p,\rho}\stackrel{g}{\to}I_{m,\tau}$ such that $s\in [\phi f\approx\psi g]$ in the sense of $(X,\sim )^{I_{p,\rho}}$.\end{enumerate}\end{defi}
\begin{rem} The third clause corrects an inaccuracy in Definition 2.5 of \cite{OostenJ:nothet}, which was noted in \cite{FruminD:homtmf}.\end{rem}
With this definition of path object we can now mimick \cite{OostenJ:nothet} and prove the following theorem:
\begin{thm}\label{Pinternalcat} \hfill\begin{enumerate}
\item The construction of ${\sf P}(X,\sim )$ extends to an endofunctor on ${\sf RT}({\mc S})$, which preserves finite limits.
\item The object ${\sf P}(X,\sim)$ comes equipped with well-defined maps:\begin{itemize}
\item[] ${\sf s}$ (source), $\sf t$ (target): ${\sf P}(X,\sim )\to (X,\sim )$;
\item[] $\sf c$ (constant path): $(X,\sim )\to {\sf P}(X,\sim )$;
\item[] $\ast$ (composition of paths): ${\sf P}(X,\sim )\times _{(X,\sim )}{\sf P}(X,\sim )\to {\sf P}(X,\sim )$ where the domain is the pullback $\{ (x,y)\, |\, {\sf t}(x)={\sf s}(y)\}$;
\item[] $\widetilde{(-)}$ (converse path): ${\sf P}(X,\sim )\to {\sf P}(X,\sim )$.\end{itemize}
Moreover, with these data the object ${\sf P}(X,\sim )$ has the structure of an internal category in ${\sf RT}({\mc S})$ with a contravariant involution which is the identity on objects.\end{enumerate}\end{thm}
\begin{prop}\label{odrefl=pathcomp} The set of path components of $(X,\sim )$, i.e.\ the quotient of $(X,\sim )$ by the equivalence relation ``there is a path from $x$ to $y$'', is exactly the order-discrete reflection of $(X,\sim )$.\end{prop}
\begin{proof} Inspection of Remark~\ref{odreflection} and Definition~\ref{pathobjectdef}.\end{proof}

Just as in \cite{OostenJ:nothet}, we see that we have a ``path object category'' in the sense of Van den Berg and Garner (\cite{vandenBergB:topsmi}).
\begin{exa}\label{topexample} For a countably based $T_0$-space, its image under the embedding in ${\sf RT}({\mc S})$ is always discrete (that is the $T_0$-property); it is order-discrete if and only if the space is $T_1$. For any $T_0$-space $X$, the {\em specialization order\/} on $X$ is the partial order defined by: $x\leq y$ if and only if every open set which contains $x$, also contains $y$ (in the notation of the introduction: $e(x)\subseteq e(y)$, if the embedding of X is $(X, E_X)$ with $E_X(x)=\{ e(x)\}$). We can see this order relation as an undirected graph on $X$. A path from $x$ to $y$ is just a path in this graph.\end{exa}

\begin{small}
\bibliographystyle{plain}
 
\end{small}

\end{document}